\documentclass[a4paper,UKenglish,cleveref, autoref, thm-restate, numberwithinsect]{lipics-v2021}

\pdfoutput=1 
\hideLIPIcs  

\graphicspath{{./figures/}}

\title{Balanced Independent and Dominating Sets \texorpdfstring{\\}{} on Colored Interval Graphs}

\author{Sujoy Bhore}{Indian Institute of Science Education and Research, Bhopal, India}{sujoy.bhore@gmail.com}{0000-0003-0104-1659}{}
\author{Jan-Henrik Haunert}{Geoinformation Group, University of Bonn, Bonn, Germany}{haunert@igg.uni-bonn.de}{0000-0001-8005-943X}{}
\author{Fabian Klute}{Algorithms and Complexity Group, TU Wien, Vienna, Austria}{fklute@ac.tuwien.ac.at}{0000-0002-7791-3604}{}
\author{Guangping Li}{Algorithms and Complexity Group, TU Wien, Vienna, Austria}{guangping@ac.tuwien.ac.at}{https://orcid.org/0000-0002-7966-076X}{}
\author{Martin N\"ollenburg}{Algorithms and Complexity Group, TU Wien, Vienna, Austria}{noellenburg@ac.tuwien.ac.at}{https://orcid.org/0000-0003-0454-3937}{}
\funding{This work was supported by the Austrian Science Fund (FWF) under grant P31119.}
\authorrunning{S.\,Bhore, J.\,Haunert, F.\,Klute,  G.\,Li, M.\,N\"ollenburg} 

\Copyright{Sujoy Bhore, Jan-Henrik Haunert, Fabian Klute, Guangping Li and Martin N\"ollenburg} 



\category{} 

\relatedversion{} 

\nolinenumbers 

\EventEditors{John Q. Open and Joan R. Access}
\EventNoEds{2}
\EventLongTitle{42nd Conference on Very Important Topics (CVIT 2016)}
\EventShortTitle{CVIT 2016}
\EventAcronym{CVIT}
\EventYear{2016}
\EventDate{December 24--27, 2016}
\EventLocation{Little Whinging, United Kingdom}
\EventLogo{}
\SeriesVolume{42}
\ArticleNo{23}
\usepackage{hyperref}
\usepackage{graphicx}
\usepackage{amsthm,amsmath,amssymb}
\usepackage{thmtools,thm-restate}
\usepackage{complexity}
\usepackage{xspace}

\newcommand{\BIS}{\textsc{BIS}\xspace}
\newcommand{\fBIS}{\textsc{$f$-BIS}\xspace}
\newcommand{\fBDS}{\textsc{$f$-BDS}\xspace}
\DeclareMathOperator{\coloring}{\gamma}
\newcommand{\fSAT}[1]{\textsc{$#1$SAT}\xspace}
\newcommand{\new}[1]{\textcolor{black}{#1}}
\newcommand{\remove}[1]{{}}

\begin{document}
\maketitle
\begin{abstract}
We study two new versions of independent and dominating set problems on vertex-colored interval graphs, namely \emph{$f$-Balanced Independent Set} ($f$-BIS) and \emph{$f$-Balanced Dominating Set} ($f$-BDS). 
 Let $G=(V,E)$ be a vertex-colored interval graph with a $k$-coloring $\coloring \colon V \rightarrow \{1,\ldots,k\}$ for some $k \in \mathbb N$.
A subset of vertices $S\subseteq V$ is called \emph{$f$-balanced} if $S$ contains $f$ vertices from each color class. 
In the $f$-BIS and $f$-BDS problems, the objective is to compute an independent set or a dominating set that is $f$-balanced. 
We show that both problems are \NP-complete even on proper interval graphs. 
\new{
For the $f$-BIS problem we design two \FPT\ algorithms, one parameterized by $(f,k)$ for interval graphs and the other parameterized by the vertex cover number for general graphs.}
Moreover, for the optimization variant of \textsc{$1$-BIS}\xspace on interval graphs, we show that a simple greedy approach achieves approximation ratio $2$.
\remove{This greedy approach also provides a constant approximation ratio for axis-parallel rectangles with bounded width or bounded height.}
\end{abstract}

	\section{Introduction}
	A graph $G$ is an interval graph if it has an intersection model consisting of intervals on the real line. 
	Formally, $G=(V,E)$ is an interval graph if there is an assignment of an interval $I_v \subseteq \mathbb R$ for each $v \in V$ such that $I_u\cap I_v$ is nonempty if and only if $\{u,v\}\in E$. 
	A \emph{proper} interval graph is an interval graph that has an intersection model in which no interval properly contains another \cite{golumbic2004algorithmic}. 
	Consider an interval graph $G=(V,E)$ and additionally assume that the vertices of $G$ are $k$-colored by a \emph{color assignment}\footnote{We use the term \emph{color assignment} instead of \emph{vertex coloring} to avoid any confusion with the general notion of vertex coloring; in particular, a color assignment $\coloring$ can map adjacent vertices to the same color.} $\coloring \colon V \rightarrow \{1,\ldots,k\}$. 
	We define and study \emph{color-balanced} versions of two classical graph problems: maximum independent set and minimum dominating set on vertex-colored (proper) interval graphs. In what follows, we define the problems formally and discuss their underlying motivation. 
	
\medskip
	 
	\noindent\underline{\textbf{$f$-Balanced Independent Set (\fBIS):}} 
	Let $G=(V,E)$ be an interval graph with a color assignment of the vertices $\coloring \colon V \rightarrow \{1,\ldots,k\}$. Find an 
	\emph{$f$-balanced independent set} of $G$, i.e., an independent set $L\subseteq V$ 
	that contains exactly $f$ elements from each color class. 
	
	\medskip

	The classic maximum independent set problem serves as a natural model for many real-life optimization problems and finds applications across fields, e.g., computer vision~\cite{balas1986finding}, information retrieval~\cite{pardalos1994maximum}, and scheduling~\cite{van2015interval}. Specifically, it has been used widely in map-labeling problems \cite{AgarwalKS98,wagner1998combinatorial,van1999point,DBLP:journals/comgeo/BeenNPW10}, where an independent set of a given set of label candidates corresponds to a conflict-free and hence legible set of labels. To display as much relevant information as possible, one usually aims at maximizing the cardinality or, in the case of weighted label candidates, the total weight of the independent set.
	This approach may be appropriate if all labels represent objects of the same category.
	In the case of multiple categories, however, maximizing the cardinality or total weight of the labeling does not reflect the aim of selecting a good mixture of different object types. For example, if the aim was to inform a map user about different possible activities in the user's vicinity, labeling one cinema, one theater, and one museum may be better than labeling four cinemas.
	In such a setting, the \fBIS problem asks for an independent set that contains $f$ vertices from each object type.

	We initiate this study for interval graphs, which is a primary step to understand the behavior of this problem on general intersection graphs. Moreover, solving the problem for interval graphs gives rise to optimal solutions for certain labeling models, e.g., if every label candidate is a rectangle that is placed at a fixed position on the boundary of the map \cite{HaunertHermes2014}.

    We also introduce an optimization variant of the \textsc{$1$-BIS} problem, which asks for a maximally colorful independent set.
	
	\medskip 
	
		\noindent\underline{\textbf{$1$-Max-Colored Independent Set (\textsc{$1$-MCIS}\xspace):}} 
	{Let $G=(V,E)$ be an interval graph with a color assignment of the vertices $\coloring \colon V \rightarrow \{1,\ldots,k\}$. The objective is to find a $1$-max-colored independent set of $G$, i.e., an independent set $L\subseteq V$, whose vertices contain a maximum number of colors and $L$ contains at most one element from each color class.
	}

\medskip 

The second problem we discuss is defined as follows.

\medskip 
	
	\noindent\underline{\textbf{$f$-Balanced Dominating Set (\fBDS):}}
	Let $G=(V,E)$ be an interval graphs with a color assignment of the vertices $\coloring \colon V \rightarrow \{1,\ldots,k\}$. Find an  \emph{$f$-balanced dominating set}, i.e., a subset $D\subseteq V$ such that every vertex in $V\setminus D$ is adjacent to at least one vertex in $D$, and $D$ contains exactly $f$ elements from each color class.
	
	\medskip 
	
The dominating set problem is another fundamental problem in theoretical computer science, which also finds applications in various fields of science and engineering~\cite{chang1998algorithmic, haynes2013fundamentals}.
Several variants of the dominating set problem have been considered over the years: $k$-tuple dominating set~\cite{chellali2012k}, Liar's dominating set~\cite{banerjee2019algorithm}, independent dominating set~\cite{irving1991approximating}, and more. 	
The colored variant of the dominating set problem has been considered in parameterized complexity, namely, red-blue dominating set, where the objective is to choose a dominating set from one color class that dominates the other color class~\cite{garnero2017linear}.
Instead, our \fBDS problem asks for a dominating set of a vertex-colored graph that contains $f$ vertices of each color class.  
Similar to the independent set problem, we primarily study this problem on vertex-colored interval graphs, which can be of independent interest. 
	
\new{
\subparagraph{Related Work}
In the field of interval scheduling~\cite{DBLP:journals/ior/CarterT92}, the colored variant of the maximum independent set problem on interval graphs has been well studied. 
This colored interval scheduling problem is known  as the \textsc{Job Interval Scheduling Problem} (\textsc{JISP}), introduced by Nakajima and Hakimi~\cite{DBLP:journals/jal/NakajimaH82}.
Let  $T = \{t_1, t_2, \ldots t_n\}$ be a set of $n$ independent tasks.
Each task $t_i$ has a task duration $d_i$ and a set of $k_i$ possible starting times $\{s_{i1}, s_{i2}, \ldots s_{ik_i}\}$.
Note that all duration and starting times of all tasks are assumed to be integers. 
The problem \textsc{JISP} is to find a starting time for each task so that each task can be executed by a single processor.
If there is some fixed integer $k$ such that  $k_i \leq k$ for $i \in \{1, 2, \ldots n\}$, the problem is referred to as \textsc{JISP}\emph{k} for short. 
Constructing a colored interval graph, where each possible execution of a task $t_i$ starting at $s_{ij}$  corresponds to a half-open interval $[s_{ij}, s_{ij} + d_i)$ of color $i$, the problem \textsc{JISP} is equivalent to finding a  \textsc{$1$-\BIS} in the constructed interval graph. 
Spieksma and Crama~\cite{spieksma1992complexity} showed that \textsc{JISP}$3$ is $\NP$-complete even if each duration $d_i$ equals $2$.
Spieksma~\cite{DBLP:conf/approx/Spieksma98} proved that the optimization variant of \textsc{JISP}, namely finding an optimal schedule that finishes a maximum number of jobs, is $\APX$-hard, and provided a $2$-approximation algorithm.
}

\new{
\subparagraph{Paper Structure}	
While there exist polynomial-time algorithms for the maximum independent set problem and the minimum dominating set problem on interval graphs~\cite{DBLP:journals/networks/GuptaLL82, DBLP:journals/siamcomp/Chang98}, it turns out their colored variants \fBIS and \fBDS are much more resilient and \NP-complete even for proper interval graphs and $f=1$ (Section~\ref{sec:bis_hardness}).
Then, we restrict our attention to the \fBIS problem.
In Section~\ref{sec:bis_algorithms}, we complement the complexity result of the \fBIS problem with two
\FPT\ algorithms, one for interval graphs and parameterized by $(f,k)$ and the other parameterized by the graph's vertex cover number and for general graphs.
Section~\ref{sec:approx-bis} introduces an $O(n \log n)$-time $2$-approximation algorithm for the optimization problem \textsc{$1$-MCIS}\xspace on interval graphs.
}
	
	\section{Complexity Results}\label{sec:bis_hardness}
\new{	In this section we show that \fBIS and \fBDS are \NP-complete even if the given graph $G$ is a proper interval graph and $f = 1$.
	It is readily seen that \fBIS and \fBDS are in the class $\NP$, since for a given potential solution $V' \subset V$ it can be checked in polynomial time whether $V'$ is an independent (dominating) set and contains $f$ vertices of each color.
	To show \NP-hardness of these two problems, we give reductions from two restricted, but still \NP-complete versions of \fSAT{3}, respectively.
	In the following, for two problems $A$ and~$B$, we write~$A \le_p B$ if there is a polynomial-time reduction from~$A$ to~$B$. }

	\subsection{Complexity of \texorpdfstring{$f$}{f}-Balanced Independent Set}\label{sub:independent_hardness}
	In a 3-bounded \fSAT{3} formula, each variable is allowed to appear in at most three clauses and clauses have two or three literals.
	The 3-bounded \fSAT{3} problem is \NP-complete~\cite{DBLP:journals/dam/Tovey84}.
This section is devoted to showing \NP-completeness of \textsc{$1$-BIS}\xspace by giving a reduction from 3-bounded \fSAT{3}.

\begin{lemma}
    \label{lem:1Reduction}
      $3$-bounded \fSAT{3}~$\le_p$ \textsc{$1$-BIS}. 
\end{lemma}
\begin{proof}
Let $\phi$ be a 3-bounded \fSAT{3} formula with
	variables $x_1,\ldots,x_n$ and clause set $\mathcal C = \{C_1,\ldots,C_m\}$. 
	From $\phi$ we now construct an equivalent instance of \textsc{$1$-BIS} consisting of a proper interval graph $G = (V,E)$ and 
	a color assignment $\coloring$ of $V$.
	We choose the set of colors to contain exactly $m$ colors,
	one for each clause in $\mathcal C$ and we number these colors from $1$ to $m$.
	We add a vertex $u_{i,j} \in V$ for each occurrence of 
	a variable $x_i$ in a clause $C_j$ in $\phi$.
	Furthermore, we insert an edge $\{u_{i,j}, u_{i,j'}\} \in E$ whenever $u_{i,j}$ was inserted because of a positive occurrence of $x_i$ and
	$u_{i,j'}$ was inserted because of a negative occurrence of $x_i$.
	Finally, we color each vertex $u_{i,j} \in V$ with color $j$.
	See Figure~\ref{fig:reduction_indset} for an illustration.
	The graph $G$ created from $\phi$ is a proper interval graph as it consists only of disjoint paths of length at most three. It can clearly be constructed in polynomial time and space.	
	\begin{figure}[htbp]
		\centering
		\includegraphics{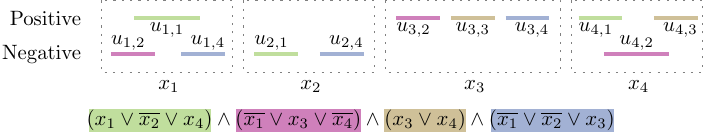}
		\caption{The constructed \textsc{$1$-BIS} instance depicted as interval representation with the vertex colors being the colors of the intervals. 
		}
		\label{fig:reduction_indset}
	\end{figure}

	It remains to prove that $\phi$ is satisfiable if and only if~$G = (V,E)$ has a \textsc{$1$-BIS} with respect to the color assignment $\coloring$.

	Suppose $V' \subseteq V$ is a $1$-balanced independent set  of $G$. 
	We construct a variable assignment for $x_1,\ldots,x_n$ as follows.
	By definition we find for each color $j$ precisely one vertex $u_{i,j} \in V'$.
	If $u_{i,j}$ was inserted for a positive occurrence of $x_i$, then we set $x_i$ to \emph{true} and otherwise $x_i$ to \emph{false}.
	Moreover, all variables $x_i$ with $i \in \{1,\ldots,n\}$ 
	for which we do not find a corresponding interval in $V'$ are also set to $\emph{false}$.
	Since $V'$ is an independent set in $G$ this assignment 
	is well defined.
	Now assume it was not satisfying, 
	then there exists a clause $C_j$ for which none of its literals evaluates to \emph{true}.
	Hence, none of the at most three vertices corresponding to the literals in $C_j$ is in $V'$.
	Recall that there is a one-to-one correspondence between clauses and colors in the instance of $1$-balanced independent set we created.
	Yet, $V'$ does not contain a vertex of that color, a contradiction.

		For the opposite direction assume we are given a satisfying assignment of the 
		3-bounded \fSAT{3}\ formula $\phi(x_1,\ldots,x_n)$ with clauses $\mathcal C = \{C_1,\ldots,C_m\}$.
		Furthermore let $ G = (V,E) $ be the graph with a color assignment of the vertices $\gamma$
		constructed from $\phi$ as described above.
		We find a $1$-balanced independent set of $G$ from the given assignment as follows.
		For each clause $C_j \in \mathcal C$ we choose one of its literals that evaluates to \emph{true} and
		add the corresponding vertex $v \in V$ to the set of vertices $V'$.
		Since there is a one-to-one correspondence between the colors and the clauses and
		the assignment is satisfying, $V'$ clearly contains one vertex per color.
		It remains to show that $V'$ is an independent set of vertices in $G$.
		Assume for contradiction that there are two vertices $v_{i,j}, v_{i',j'} \in V'$ and
		$\{v_{i,j},v_{i',j'}\} \in E$. Then, by construction of $G$, we know that $i = i'$ and further that
		$v_{i,j}, v_{i',j'}$ correspond to one positive and one negative occurrence of $x_i$ in $\phi$.
		By the construction of $V'$ this implies a contradiction to the assignment being consistent.
\end{proof}		
\new{
Evidently \textsc{$1$-BIS} is in the class $\NP$.
Together with Lemma~\ref{lem:1Reduction}, this implies the following theorem.}
	\begin{restatable}{theorem}{thmreductionindset}
		\label{thm:reduction_indset}
		The $f$-balanced independent set problem on a graph $G = (V,E)$ with a color assignment of the vertices $\coloring \colon V \rightarrow \{1,\ldots,k\}$ is \NP-complete, even if $G$ is a proper interval graph and $f=1$.
	\end{restatable}

\begin{remark}
\new{
Since the constructed graph $G$ consists of a set of disjoint paths of length at most three, 
it is easy to construct a set $I$ of right-open intervals of length $2$ such that $G$ represents the intersections of $I$.
Moreover, each color class consists of at most three intervals. 
Therefore, we independently prove the $\NP$-completeness of the \textsc{Job Interval Scheduling Problem} even when each interval has length $2$ and at most three intervals are of the same color, as stated in~\cite{spieksma1992complexity}.
}
\end{remark}

	\subsection{Complexity of \texorpdfstring{$f$}{f}-Balanced Dominating Set}\label{dominating-bis}
	
	In a \textsc{2P2N}-\fSAT{3} formula, each variable appears exactly twice positively and twice negatively. 
	The \textsc{2P2N}-\fSAT{3} problem is \NP-complete~\cite{DBLP:conf/rta/Yoshinaka05}.
	In this section, we show \NP-hardness of \textsc{$1$-BDS} by a reduction from \textsc{2P2N}-\fSAT{3}.
	
\subparagraph{Construction}	
	Let $\phi$ be a \textsc{2P2N}-\fSAT{3} formula with
	variables $x_1,\ldots,x_n$ and clause set $\mathcal C = \{C_1,\ldots,C_m\}$.
	For variable $x_i$ in $\phi$ we denote with $\mathcal C_{x_i} = \{C_t^1, C_t^2, C_f^1, C_f^2\}$
	the four clauses $x_i$ appears in, where $C_t^1,C_t^2$ are clauses with positive occurrences of $x_i$ and $C_f^1,C_f^2$ are clauses containing negative occurrences of $x_i$.
	We construct a vertex-colored graph $G = (V,E)$ from 
	$\phi$ as follows.
	For each variable $x_i$ we introduce six vertices $t_1,t_2,f_1,f_2,h_t,$ and $h_f$ and
	for each clause $C_j$ with occurrences of variables $x_{j_1}$, $x_{j_2}$, and $x_{j_3}$ 
	we add up to three vertices $c_k$ for each $k \in \{j_1,j_2,j_3\}$. (In case a clause has less than three literals we add only one or two vertices.)
	If the connection to the variable is clear, we also write $c_t^1$, $c_t^2$, $c_f^1$, and $c_f^2$
	for the vertices introduced for this variable's occurrences in the clauses 
	$C_t^1, C_t^2, C_f^1,$ and $C_f^2$, respectively.
	Furthermore, we add for each variable $x_i$ the edges 
	$\{h_t,t_1\}$, $\{h_t,t_2\}$, $\{h_f, f_1\}$, and $\{h_f,f_2\}$, as well as
	for each clause $C_j$ all possible edges between the three vertices introduced for $C_j$.
	For each variable $x_i$ we introduce five colors, namely
	$z_t^1$, $z_t^2$, $z_f^1$, $z_f^2$, and $z_h$.
	We set $\coloring(h_t) = \coloring(h_f) = z_h$.
	Finally, we set $\coloring(t_1) = \coloring(c_t^1) = z_t^1$.
	Equivalently for $t_2$, $f_1$, and $f_2$.
	See Figure~\ref{fig:reduction_domset} for an example.
	In total we create $6n + 3m$ many vertices and $4n + 3m$ many edges, thus the reduction is polynomial. 
	All variable and clause gadgets are independent components and 
	only consist of paths of length three and triangles, hence $G$ is a proper interval graph.
		\begin{figure}[htbp]
		\centering
		\includegraphics{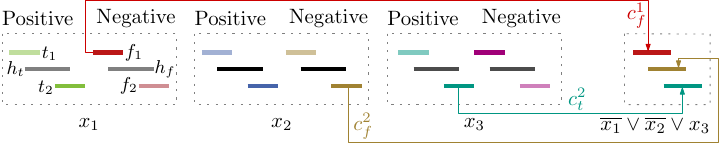}
		\caption{Illustrations of three variable gadgets and a clause gadget as interval representations.}
		\label{fig:reduction_domset}
	\end{figure}

	To establish the correctness of our reduction for 1-BDS we first introduce a canonical type of 
	solutions for the graphs produced by our reduction.
	We call $V_D$ \emph{canonical}, if for each variable $x_i$ we either find 
	$\{h_t,f_1,f_2,c_t^1,c_t^2\} \subset V_D$ or $\{h_f,t_1,t_2,c_f^1,c_f^2\} \subset V_D$. 
	If for a variable $x \in X$ and a $1$-balanced dominating set $V_D \subseteq V$ we find one of the two above sets in $V_D$, we say $x$ is in \emph{canonical form} in $V_D$. 
	The next lemma shows that if $G$ has a $1$-balanced dominating set we can turn it into a canonical one.
	
	\begin{restatable}{lemma}{lemcanonical}
		
		\label{lem:canonical}
		Let $G = (V,E)$ be a graph generated from a \textsc{2P2N}-\fSAT{3} formula $\phi$ with clause set $\mathcal C = \{C_1,\ldots,C_m\}$ as above and $V_D \subseteq V$ a $1$-balanced dominating set, then $V_D$ can be transformed into a canonical $1$-balanced dominating set in $O(|V|)$ time.
	\end{restatable}
	\begin{proof}
		Let $x$ be not in canonical form in $V_D$. Since $V_D$ is a $1$-balanced dominating set we know that either $h_t$ or $h_f$ of $x$ is in $V_D$. 
		Without loss of generality assume that $h_t \in V_D$.
		Consequently, we find that $f_1,f_2 \in V_D$ and $c_f^1,c_f^2\not\in V_D$.
		Now, we obtain the set $V_D'$ from $V_D$ by 
		removing any occurrence of $t_1$ or $t_2$ from $V_D$ and
		inserting all missing elements of $\{c_t^1,c_t^2\}$.

		Clearly $x$ is in canonical form in $V_D'$.
		We need to show that $V_D'$ is still a $1$-balanced dominating set.
		It is straight forward to verify that every color appears exactly once in $V_D'$ 
		if $V_D$ was $1$-balanced.
		Now assume there was a vertex $u \in V$ that is not dominated by any vertex in $V_D'$.
		Yet, we at most deleted $t_1$ and $t_2$ in $V_D'$ but since $h_t \in V_D'$ 
		both and all their neighbors are dominated.
		As our operations only affected vertices introduced for $x$ and occurrences of $x$
		we can simply iterate this process for each variable until every variable $x_i$ is in canonical form.
	\end{proof}

	\begin{restatable}{theorem}{thmreductiondomset}
		
		\label{thm:reduction_domset}
		The $f$-balanced dominating set problem on a graph $G = (V,E)$ with a color assignment of the vertices $\coloring : V \rightarrow \{1,\ldots,k\}$ is \NP-complete, even if $G$ is a proper interval graph and $f=1$.
	\end{restatable}
	\begin{proof}
		
		Let $G = (V,E)$ be constructed from a \textsc{2P2N}-\fSAT{3} formula $\phi$ with clause set $\mathcal C = \{C_1,\ldots,C_m\}$ as above and 
		let $V_D$ be a $1$-balanced dominating set of $G$. 
		By Lemma~\ref{lem:canonical} we can assume $V_D$ is canonical. 
		We construct an assignment of the variables in $\phi$ by 
		setting $x_i$ to \emph{true} if its $h_t \in V_D$ and to \emph{false} otherwise. 
		Assume this assignment was not satisfying, i.e., 
		there exists a clause $C_j \in \mathcal C$ such that none of the literals in $C_j$
		evaluates to \emph{true}.
		For each positive literal of $C_j$ we then get that the corresponding variable $x_i$
		was set to \emph{false}.
		Hence, $h_f \in V_D$ for $x_i$ and consequently $c_t^1,c_t^2\not\in V_D$.
		Equivalently for each negative literal we find $h_t \in V_D$ and $c_f^1,c_f^2 \not\in V_D$.
		As a result we find that none of the vertices introduced for literals in $C_j$ is in $V_D$ and
		especially that none of them is dominated as they are each others only neighbors.
		Yet, $V_D$ is a $1$-balanced dominating set by assumption, a contradiction.
		
		In the other direction, assume we are given a satisfying assignment of a \textsc{2P2N}-\fSAT{3} formula 
		$\phi$ with clause set $\mathcal C = \{C_1,\ldots,C_m\}$.
		Furthermore, let $G = (V,E)$ be the graph constructed from $\phi$ as above.
		We form a canonical $1$-balanced dominating set $V_D \subseteq V$ of $G$ in the following way. 
		For every variable $x_i$ that is set to \emph{true} in the assignment 
		we add $\{h_t,f_1,f_2,c_t^1,c_t^2\}$ to $V_D$ and 
		for every variable $x_{i'}$ that is set to \emph{false} we add $\{h_f,t_1,t_2,c_f^1,c_f^2\}$. 
		This clearly is a $1$-balanced set and it is canonical. 
		It remains to argue that it dominates $G$. 
		For the vertices introduced for variables this is clear, 
		since we pick either $h_t$ or $h_f$, as well as $f_1,f_2$ or $t_1,t_2$ for every variable $x_i$. 
		Now, assume there was a clause $C_j \in \mathcal C$ and none of the vertices introduced for literals in $C_j$ was in $V_D$. 
		Then, by construction of $V_D$, we find that 
		for any positive (negative) occurrence of a variable $x_i$ in $C_j$ 
		the variable $x_i$ was set to \emph{false} (\emph{true}).
		A contradiction to the assignment being satisfying.
	\end{proof}

	\section{Algorithmic Results for the Balanced Independent Set }\label{sec:bis_algorithms}
	
	
	In this section, we take a parameterized perspective on \fBIS and provide two \FPT\ algorithms with different parameters.
	The algorithms described in this section can be easily generalized to maximize the value of $f$ in \fBIS.

	\subsection{An FPT Algorithm For Interval Graphs Parameterized by (\texorpdfstring{$f, k$}{f,k})}\label{xp-bis} 
	Assume we are given an instance of \fBIS with
	$G=(V,E)$ being an interval graph with a color assignment of the vertices  $\coloring \colon V \rightarrow \{1,\ldots,k\}$.
	We can construct an interval representation $\mathcal I = \{I_1,\ldots,I_n\}$, $n = |V|$,
	from $G$ in linear time \cite{Lekkeikerker1962}. 
	Our algorithm is a dynamic programming based procedure that work as follows. 
	
	Firstly, we sort the right end-points of the $n$ intervals in $\mathcal I$ in ascending order.
	Next, we define a function $ prev \colon \mathcal I \rightarrow \{1,\ldots,n\}$ such that for each interval $I_{i} \in \mathcal I$, $prev(I_i)$ is 
	the index of the rightmost interval with its
	right endpoint left of $I_i$'s left endpoint.
	If no such interval exists for some interval $I_i$, we set $prev(I_i) = 0$.

	For each $\kappa\in \{1,\ldots,k\}$, let 
	$\hat{e}_{\kappa}$ denote the $k$-dimensional 
	unit vector of the form $(0, \ldots, 0, 1, 0, \ldots, 0)$, 
	where the element at the $\kappa$-th position is $1$ and the rest are $0$.
	For a subset $\mathcal I' \subseteq \mathcal I$ we define a \emph{cardinality vector} as the 
	$k$-dimensional vector $C_{\mathcal I'} = (c_1,\ldots,c_k)$, 
	where each element $c_i$ represents 
	the number of intervals of color $i$ in $\mathcal I'$. 
	We say $C_{\mathcal I'}$ is \emph{valid} if all $c_i \leq f$ and
	the set $\mathcal I'$ is independent.
	
	The key observation here is that there are at most $O((f+1)^k)$ many different valid cardinality vectors as there are only $k$ colors and we are interested in at most $f$ intervals per color.
	In the following let $U_j$, $j\in \{1,\ldots,n\}$, be the union of 
	all valid cardinality vectors of the first $j$ intervals in $\mathcal I$.
    \new{Furthermore, we store for each $u \in U_{j}$ one representative interval set with $u$ as its cardinality vector.}

	Let $U_0 = \{(0,\ldots,0)\}$ in the beginning.
\new{	To compute an $f$-balanced independent set the algorithm simply iterates over all right endpoints of the intervals in $\mathcal I$ and in the $i$-th step computes $U_i$ as the set consisting of all valid cardinality vectors of the union $\{ u + \hat{e}_{\coloring(I_i)} \mid u\in U_{prev(I_i)}\} \cup U_{i-1}$.
	Correspondingly, we update the representative interval sets for $U_{i}$.}
	Finally, we check the cardinality vectors in $U_n$ and return \emph{true} in case there is one cardinality vector $w \in U_n$
	with entries being all $f$ and \emph{false} otherwise. Moreover, the representative interval set of $w$ is an $f$-balanced independent set.

	\begin{restatable}{theorem}{thmxpalgo}\label{xp-algo}
		Let $G=(V,E)$ be an interval graph with a color assignment of the vertices  $\coloring : V \rightarrow \{1,\ldots,k\}$.
		We can compute an $f$-balanced independent set of $G$ or
		determine that no such set exists in $O(n \log n + k (f+1)^k n)$ time. 
	\end{restatable}
	\begin{proof}
		Let $\mathcal I = \{I_1,\ldots,I_n\}$ be an interval representation of $G$ on which we execute our algorithm.
		For $U_0$ the set just contains the valid cardinality vector with all zeros
		which is clearly correct.
		Let $U_{i-1}$ be the set of valid cardinality vectors computed after step $i - 1$.
		Now, in step $i \leq n$ we calculate the set $U_i$ as the union of $U_{i-1}$ and
		the potential new solutions based on independent sets of intervals containing $I_i$.
		Assume $\mathcal I_x \subseteq \{I_1,\ldots,I_i\}$ is an independent set of intervals such that 
		its cardinality vector $C_{\mathcal I_x}$ is valid, but there is
		no valid cardinality vector $C_{\mathcal I'} \in U_i$ such that $C_{\mathcal I'}$ is larger or equal in every component than $C_{\mathcal I_x}$.
		Since $U_{i-1}$ contained all valid cardinality vectors for the intervals in $\{I_1,\ldots,I_{i-1}\}$ 
		we know that $C_{\mathcal I_x}$ is such that $I_i \in \mathcal I_x$.
		Yet, the set $U_{prev(I_i)}$ contained all valid cardinality vectors 
		for the set of intervals $\{I_1,\ldots,I_{prev(I_i)}\}$.
		Since $I_i$ has overlaps with all intervals in \{$I_{prev(I_i) + 1},\ldots,I_{i-1}\}$ and 
		hence cannot be in any independent set with any such interval we can conclude that
		$C_{\mathcal I_x} - \hat{e}_{\gamma(I_i)} \in U_{prev(I_i)}$.
		Thus, we also find $C_{\mathcal I_x} \in U_i$, a contradiction.
		
		Next we consider the running time. 
		The key observation is that there are at most $(f+1)^k$ different valid cardinality vectors.
		Checking the validity can be done in $O(1)$ time for each new vector as only one entry changes.
		Computing the sets $U_i$ can be done in time $O(k(f+1)^k)$,
		by storing the cardinality vectors in lexicographically sorted order for each set.
		Keeping the sets in sorted order does not require any extra running time, 
		as $U_0$ is clearly sorted in the beginning (it only contains one element) and
		we only increase the same entry for each vector in $U_{prev(I_i)}$ when forming the union, 
		thus not changing their ordering.
		Hence, the set $U_{prev(I_i)}$  and $U_{i-1}$ can be assumed to be sorted in lexicographic order.
		Consequently, by merging from smallest to largest element the set $U_i$ is again lexicographically sorted after the union.
		Furthermore, we can easily discard double entries by comparing also against the vector we inserted last into $U_i$.
		Finally, we have to sort the intervals themselves.
		Using standard sorting algorithms this works in $O(n\log n)$ time.
		Altogether, this results in a running time of $O(n \log n +   k (f+1)^k n)$.
	\end{proof}

	\subsection{An FPT Algorithm Parameterized by the Vertex Cover Number}\label{fpt-bis}
	Here we will give an alternative \FPT\ algorithm for   \fBIS, this time parameterized by the \textit{vertex cover number} $\tau(G)$ of $G$, i.e., the size of a minimum vertex cover of $G$.
	\new{This alternative algorithm works on general graphs, not only on interval graphs.
	In the following, for a vertex set $S \subseteq V$, $N(S)$ denotes the neighborhood of $S$ consisting of all vertices adjacent to a vertex of $S$. }
	
	\begin{restatable}{lemma}{vertexcoverswap}
		\label{lem:vertex cover swap}
		Let $G = (V,E)$ be a graph. Consider a vertex cover $V_c$ in $G$ and its complement $V_{\mathrm{ind}}=V \setminus V_c$. Then any maximal independent set $M$ of $G$ can be constructed from $V_{\mathrm{ind}}$ by adding the subset $M\cap V_c$ of $V_c$  and removing its neighborhood in $V_{\mathrm{ind}}$, namely $M=(V_{\mathrm{ind}} \cup(M\cap V_c)) \setminus N(M\cap V_c)$.
	\end{restatable}
	\begin{proof} 
		For a fixed but arbitrary maximal independent set $M$, in the following, we denote the set ($V_{\mathrm{ind}} \cup(M\cap V_c))\setminus N(M\cap V_c)$ as $M_{\mathrm{swap}}$.

		We first prove the independence of $M_{\mathrm{swap}}$. Note that by the definition of a vertex cover $V_{\mathrm{ind}}$ is an independent set. 
		Furthermore, the set $(M\cap V_c)$, as a subset of the independent set $M$, is also independent. 
		Then, in the union $V_{\mathrm{ind}} \cup(M\cap V_c)$ of these two independent sets, any adjacent pair of vertices must contain one vertex in $M\cap V_c$ and one in $V_{\mathrm{ind}}$. 
		Hence, after removing all the neighboring vertices of $M\cap V_c$, the set $M_{\mathrm{swap}}$ is independent.
		
		Next we prove that $M\subseteq M_{\mathrm{swap}}$. 
		Assume there exists one vertex $v_m$ in $M$ but not in $M_{\mathrm{swap}}$. 
		Since $v_m \in M$ it must also be in the set  $V_{\mathrm{ind}} \cup(M\cap V_c)$. 
		With the assumption that $v_m \notin M_{\mathrm{swap}}$, 
		we get that $v_m$ must be in $N(M\cap V_c)$. 
		Consequently, $v_m$ is in the independent set $M$ and is at the same time a neighbor of vertices in $M$, a contradiction.
		
		Finally we prove $M = M_{\mathrm{swap}}$. 
		We showed above that $M_{\mathrm{swap}}$ is an independent set and also $M\subseteq M_{\mathrm{swap}}$. 
		Since $M$ is a maximal independent set by assumption we get $M = M_{\mathrm{swap}}$.
	\end{proof}

	\begin{restatable}
		{lemma}{MISnumber}
		\label{lem:MIS number}
		Let $G = (V,E)$ be a graph with vertex cover number $\tau(G)$. There are $O(2^{\tau(G)})$ maximal independent sets of $G$.  
	\end{restatable}
	\begin{proof}
		Consider a minimum vertex cover $V_c$ in $G$ and its complement $V_{\mathrm{ind}}=V \setminus V_c$. Note that since $V_c$ is a (minimum) vertex cover, $V_{\mathrm{ind}}$ is a (maximum) independent set. Furthermore, any maximal independent set $M$ of $G$ can be constructed from $V_\mathrm{ind}$ by adding $M\cap V_c$ 
		and removing its neighborhood in $V_{\mathrm{ind}}$, namely $M=(V_{\mathrm{ind}} \cup(M\cap V_c)) \setminus N(M\cap V_c)$ by Lemma~\ref{lem:vertex cover swap}.
		Thus there are $O(2^{\tau(G)})$ maximal independent sets of  $G$. 
	\end{proof}

	\begin{theorem}
		\label{thm:fpt algorith}
		Let $G = (V,E)$ be an interval graph with a color assignment of the vertices  $\coloring \colon V \rightarrow \{1,\ldots,k\}$. We can compute an $f$-balanced independent set of $G$ or
		determine that no such set exists in $O(2^{\tau(G)}\cdot n)$ time.
	\end{theorem}
	\begin{proof}
		According to Lemma~\ref{lem:MIS number}, there are $O(2^{\tau(G)})$ maximal independent sets of  $G$.
		The basic idea is to enumerate all the $O(2^{\tau(G)})$ maximal independent sets and compute their maximum balanced subsets. Enumerating all maximal independent sets of an interval graph takes $O(1)$ time per output \cite{DBLP:journals/jda/OkamotoUU08}. Given an arbitrary independent set of $G$ we can compute an $f$-balanced independent subset in $O(n)$ time or conclude that no such subset exists. Therefore, the running time of the algorithm is $O(2^{\tau(G)}\cdot n)$.
	\end{proof}
	
	\section{A 2-Approximation for the 1-Max-Colored Independent Set}\label{sec:approx-bis}

	We note that the \NP-completeness of 1-BIS implies that \textsc{$1$-MCIS}\xspace is an \NP-hard optimization problem.
\new{	Moreover, the $\APX$-hardness of \textsc{Job Interval Scheduling Problem}~\cite{DBLP:conf/approx/Spieksma98}  implies that \textsc{$1$-MCIS} does not have a $\PTAS$ approach as well.
	Spieksma~\cite{DBLP:conf/approx/Spieksma98} showed that a simple sweep-line algorithm can provide approximation ratio $2$ for $\textsc{JISP}$.
	Now we first review this sweep-line algorithm and show that this algorithm  provides ratio $2$ for \textsc{$1$-MCIS}\xspace problem\remove{ and then extend the result to rectangle intersection graphs}.}

	First, we sort the intervals from left to right based on 
	their right end-points. Then, our algorithm scans the intervals from left to right, and at each step selects greedily an interval of a distinct color such that no interval of the same color has been selected before. Moreover, we maintain a solution array $S$ of size~$k$ to store the selected intervals. 
	
	For each interval $I_i$ in this order, we check if the color of $I_i$ is still missing in our solution (by checking if $S[\coloring(I_i)]$ is not yet occupied). If yes, we store $I_i$ in $S[\coloring(i)]$ and remove all the remaining intervals overlapping $I_i$. 
	Otherwise, if $S[\coloring(I_i)]$ is not empty, we remove $I_i$ and continue scanning the intervals. This process is repeated until all intervals are processed. 
	Then, by using a simple charging argument on the colors in an optimal solution that are missing in our greedy solution, we obtain the desired approximation factor. 
	
	\begin{figure}[htbp]
		\centering
		\includegraphics{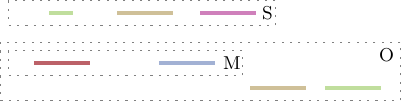} 
		\caption{Comparison of a solution $S$ of the algorithm and an optimal solution $O$. Subset $M \subseteq O$ contains two colors (red and blue) missing from $S$, but each interval in $M$ contains the right endpoint of a different interval from $S$.}
		\label{fig:approximation-example}
	\end{figure}
	
	\begin{restatable}{theorem}{thmapprox}
		\label{thm:2-apprximate}
		Let $G = (V,E)$ be an interval graph with a color assignment of the vertices $\coloring \colon V \rightarrow \{1,\ldots,k\}$. In $O(n \log n)$ time, we can compute an independent set with at least $\lceil \frac{c}{2} \rceil$ colors, where $c$ is the number of colors in a 1-max-colored independent set.
	\end{restatable}

	\begin{proof}
		It is clear from the above description that the greedy algorithm finds an independent set. 
		We maintain a solution array $S$, and it is possible to check if an interval of a particular color is already available in $S$ in constant time. 
		Therefore, the entire algorithm runs in $O(n \log n)$ time. 
		
		In order to prove the approximation factor, we compare the solution $S$ of our greedy algorithm with a fixed 1-max-colored independent set $O$ (see Figure~\ref{fig:approximation-example}).
		Let $M = \{I_i \in O \mid \nexists I_j \in S \text{ with } \coloring(I_j)=\coloring(I_i)\}$ be the subset of $O$ consisting of intervals of missing colors in $S$. 
		Now, consider an interval $I_m \in M$. There must be at least one interval $I_s\in S$, whose right endpoint is contained in the interval $I_m$. Otherwise, since there is no interval of the same color as $I_m$ in $S$, the greedy algorithm would scan $I_m$ as the interval with the leftmost right endpoint in the process and select it in $S$. Thus, the function $\rho$, which maps each interval $I_m$ in $M$ to an interval $I_s$ in $S$ such that $I_s$ is the rightmost interval in $S$ with its right endpoint is contained in $I_m$, is well-defined. Furthermore, $\rho$ is an injective function because of the independence of the set $M$. Therefore, we can conclude that the cardinality of the set $S$ is greater than or equal to the  cardinality of  $M$. Note that, $|M|+|S| \geq |O|$. Hence, $S$ has size at least $\lceil \frac{c}{2} \rceil$.
	\end{proof}
\remove{	
In the following, we consider axis-aligned rectangles of bounded width or bounded height.
Let  $\mathcal R$ be a set of axis-parallel rectangles of bounded width and arbitrary height; the width of each rectangle of $\mathcal{R}$ is at least 1 unit and bounded by $b_w$ for a fixed number $b_w$.
 Let $\leq_{lex}$ denote the \emph{lexicographical ordering} of the points in the plane. For two points $p_1 =(x_1, y_1)$ and $p_2 =(x_2, y_2)$, $p_1 \leq_{lex} p_2$ if $y_1 < y_2$ or  $y_1 = y_2 \land x_1 < x_2$.
We assume 
that the rectangles in  set $R=\{r_1, r_2, \ldots r_n\}$ are in a general position such that each rectangle of $\mathcal R$ has unique corners. 
We also assume without loss of generality that rectangles of $\mathcal R$ are sorted in a decreasing order by their bottom left corner points lexicographically.} 

\remove{We get the following observation by a simple packing argument; see Figure~\ref{fig:packing}.
 \begin{observation}
 \label{obs:independency packing}
Let $S\subset \mathcal R$ be a subset of $\mathcal R$ consisting of 
pairwise overlap-free rectangles and let $r\in \mathcal R$ be an arbitrary rectangle of $\mathcal R$.
There are at most $(\lceil b_w \rceil +1)$ rectangles of $\mathcal R$ that have a greater index than $r$ and intersect $r$.
\end{observation}
\begin{figure}[htbp]
    \centering
    \includegraphics{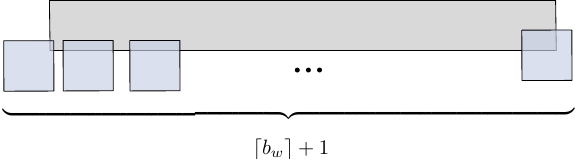}
    \caption{An illustration of Observation~\ref{obs:independency packing}. At most $(\lceil b_w \rceil +1)$ rectangles (unit blue squares) can be placed on the bottom edge of the  rectangle $r$ (gray).}
    \label{fig:packing}
\end{figure}
With the same idea that we used for the interval graphs, we apply the simple greedy  approach on $\mathcal R$
by sweeping the rectangles from top to bottom in the ordering of $\mathcal R$ and greedily extend an independent set by adding new rectangles of non-collected color classes.
It is not hard to see that this approach finds an independent set of $\mathcal R$ in $O(n \log\, n)$ time (including the time of ordering rectangles of $\mathcal R$ lexicographically).
Now we prove that this simple greedy sweep-line approach reaches a approximation ratio $\lceil b_w \rceil +2$ for the \textsc{$1$-MCIS}\xspace problem.
\begin{theorem}
Let $G = (V,E)$ be a rectangle intersection graph induced by a set of rectangles $\mathcal R$ with width bound $b_w$ and let $\coloring \colon V \rightarrow \{1,\ldots,k\}$ be a color assignment of $V$.
The independent set computed by the greedy sweep-line approach contains  at least $\lceil \frac{c}{\lceil b_w \rceil +2} \rceil$ colors, where $c$ is the number of colors in a 1-max-colored independent set with respect to the color assignment $\coloring$.    
\end{theorem}
\begin{proof}
This proof follows the same charging argument as for the interval graph. 
Let $S$ be the independent set computed by the greedy sweep-line approach, and let $O$ be a fixed \textsc{$1$-MCIS} of $G$.
We partition $O$ into two set $M$ and $H$ such that $H$ consists of all intervals whose colors are picked in $S$ and $M$ consists of intervals of colors missing in $S$.
By definition, it is clear that $|H| \leq |S|$.
It remains to prove that $|M| \leq (\lceil b_w \rceil + 1)|S|$.
For each rectangle $m$ of $M$, there must be a rectangle $s \in S$ such that $m$ intersects $s$ and $s$ has smaller index than $m$; Otherwise, $m$ would have been picked in $S$.
Thus, the function $\rho$, which maps each rectangle $m$ of $M$ to the rectangle $s$ of $S$ which has the smallest index among all rectangle of $S$ intersecting $m$ is well-defined. 
By Observation~\ref{obs:independency packing},  for every rectangle of $S$ there are at most $\lceil b_w \rceil +1$ rectangles of $M$ intersecting $s$ with indices greater than the index of $s$.
It implies that at most $\lceil b_w \rceil +1$ rectangles of $M$ are mapped by $\rho$ to the same rectangle of $S$.
Hence, $M$ has size at most $(\lceil b_w \rceil + 1)|S|$.
Overall, we get $|O| = |H| + |M| \leq |S| + (\lceil b_w \rceil + 1)|S|$.
\end{proof}
\begin{corollary}
The greedy sweep-line approach computes a $3$-approximation solution for \textsc{$1$-MCIS} problem on unit rectangles.
\end{corollary}
}
\begin{center}
\fbox{
\begin{minipage}{0.97\textwidth}
\subparagraph{Correction to an earlier version of this paper.}\new{
In an earlier version of this work~\cite{bhkln-bidscig-21},  we presented a local search approach and claimed it being a \PTAS\xspace for the \textsc{$1$-MCIS} problem.
After its publication date, we found the paper by Spieksma~\cite{DBLP:conf/approx/Spieksma98} and noted that our claim was in contradiction with their \APX-hardness result of the \textsc{Job Interval Scheduling Problem}.
We rechecked our charging argument and found an error in our proof.}

\new{More precisely, let $\mathcal{L}$ be the solution of the local search algorithm, let $\mathcal{O}$ be an optimal solution, and let $H$ be the graph induced by $\mathcal{L}\cup \mathcal{O}$ that has one vertex for each interval and two vertices are connected by an edge if the corresponding intervals intersect or belong to the same color class.
We distinguish between these two types of edges:
the former edges are called \emph{interval-edges} and 
the latter are called \emph{color-edges}.
We claimed that ``$H$ is a planar graph'' \cite[Lemma 4]{bhkln-bidscig-20}.
This statement is incorrect; see a counterexample in the following figure.
Here, the graph $H$ induced by $\mathcal{L} \cup \mathcal{O}$ contains color-edges $(L_i, O_i)$ and $(L_{ij}, O_{ij})$ as well as interval-edges  $(L_i, O_{ij})$ and $(L_{ij},O_j)$ for $i, j \in  \{1,2,3\}$ and $i \neq j$.
The color-edges $(O_i, L_i)$
and the internally disjoint paths $(L_i, O_{ij},L_{ij},O_j)$ for $i, j \in  \{1,2,3\}$ and $i \neq j$ together yield a  $K_{3,3}$-minor of $H$, which contradicts the planarity of $H$.}

\vspace{0.5cm}

\includegraphics[width = \textwidth]{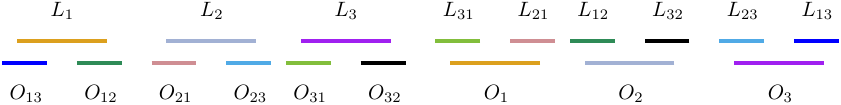}
\end{minipage} } 
\end{center}

	\section{Conclusions}
	In this paper, we have studied the $f$-Balanced Independent and Dominating
	set problem for interval graphs. We proved that these problems are
	\NP-complete and obtained algorithmic results for the $f$-Balanced
	Independent Set problem. An interesting direction for future work is to obtain algorithmic
	results for the $f$-Balanced Independent Set problem for other geometric
	intersection graphs, e.g., rectangle intersection graphs, unit disk graphs
	etc. Our results may help to tackle these problems since algorithms for
	computing (maximum weighted) independent sets of geometric objects in the
	plane often use algorithms for interval graphs as subroutines. Another
	interesting problem is to design approximation or parameterized algorithm
	for the $f$-Balanced Dominating Set problem for interval graphs.

\bibliography{main.bib}
\end{document}